\documentclass[submission,creativecommons]{eptcs}

\usepackage{amsthm}
\usepackage{amsmath}

\usepackage[pdftex]{virginialake}
\usepackage[latin1]{inputenc}
\usepackage{stmaryrd}
\usepackage{xspace}

\newcommand{\grdef}[0]{\mathrel{\raise.5pt\hbox{$\mathop{::}$}{=}}}
\newcommand{\mmid}[0]{\:\mid\:}
\newcommand{\qqquad}[0]{\qquad\quad}
\newcommand{\qqqquad}[0]{\qquad\qquad}
\newcommand{\Per}[0]{\Omega}
\newcommand{\LLM}[0]{{\sf LLM}\xspace}

\newcommand{\rn}[1]{${\mathrm{#1}}$\xspace}
\newcommand{\rnm}[1]{{{\mathrm{#1}}}\xspace}

\newcommand{\nfy}[0]{\mathord{\uparrow\hspace*{0.2pt}}}
\newcommand{\pfy}[0]{\mathord{\downarrow\hspace*{0.4pt}}}
\newcommand{\activ}[0]{\nnearrow\xspace}

\newcommand{\rseq}[1]{\vdash #1}
\newcommand{\fseq}[1]{\vDash #1}

\newcommand{\llo}[0]{{{\sf 1}}\xspace}
\newcommand{\llz}[0]{{{\sf 0}}\xspace}
\newcommand{\parr}[0]{\mathbin{\bindnasrepma}\xspace}
\newcommand{\tensor}{\otimes}

\newcommand{\with}[0]{\mathbin{\binampersand}\xspace}
\newcommand{\oc}[0]{\mathord{!}\xspace}
\newcommand{\wn}[0]{\mathord{?}\xspace}
\newcommand{\an}[1]{\overline{#1}}
\newcommand{\lln}[0]{^\bot}

\newcommand{\iaxrule}[2]{
  \vlinf{\rnm{#1}}{}{#2}{}}
\newcommand{\irule}[3]{
  \vlinf{\rnm{#1}}{}{#2}{#3}}

\newcommand{\iruule}[4]{
  \vliinf{\rnm{#1}}{}{#2}{#3}{#4}}

\newcommand{\ider}[1]{\vlderivation{#1}}

\newcommand{\idin}[3]{
  \vlin{\rnm{#1}}{}{#2}{#3}}

\newcommand{\idiin}[4]{
  \vliin{\rnm{#1}}{}{#2}{#3}{#4}}
\newcommand{\idax}[2]{
  \vlin{\rnm{#1}}{}{#2}{\vlhy{}}}
\newcommand{\idopen}[2]{
  \idtrees{#1}{#2}{0.1em}}
\newcommand{\idtrees}[3]{
  \vltr{#1}{#2}{\vlhy{}}{\vlhy{\hspace{#3}}}{\vlhy{}}}

\def\pf{\vdash}
\def\fpf{\vDash}
\def\tensor{\otimes}

\def\foc#1{[#1]}

\def\nshift{\pfy}
\def\pshift{\nfy}

\def\Dd{\mathcal{D}}
\def\Ee{\mathcal{E}}

\newtheorem{lemma}{Lemma}
\newtheorem{theorem}[lemma]{Theorem}
\newtheorem{example}{Example}
\def\dstep#1::#2&#3\\{
\> $#1$ \> :: \> $#2$ \` #3 \\}
\def\dstepn#1&#2\\{
\>      \>    \> $#1$ \` #2 \\}
\def\dgoal#1&#2\\{
\>      \>    \> $#1$ \` #2}
\def\dstepnl#1::#2&#3\\{
    \dstep   #1 :: #2 & {\\\` #3} \\}
\newenvironment{lineproof}{
\begin{tabbing}
\quad \= $\mathcal F_2'$ \= :: \= \kill}
{\end{tabbing}}

\def\deriv#1::#2{%
  \idtrees{#1}{#2}{1em}}

\def\delay#1{\langle#1\rangle}

\title{Cut Elimination in Multifocused Linear Logic}

\author{Taus Brock-Nannestad
\institute{INRIA \& LIX, \'Ecole Polytechnique}
\email{taus.brock-nannestad@inria.fr}
\and Nicolas Guenot
\institute{IT University of Copenhagen}
\email{ngue@itu.dk}}

\begin{document}
\vlnosmallleftlabels

\maketitle

\begin{abstract}
We study cut elimination for a multifocused variant of full linear logic
in the sequent calculus. The multifocused normal form of proofs yields
problems that do not appear in a standard focused system, related to
the constraints in grouping rule instances in focusing phases. We
show that cut elimination can be performed in a sensible way even though
the proof requires some specific lemmas to deal with multifocusing phases,
and discuss the difficulties arising with cut elimination when considering
normal forms of proofs in linear logic.
\end{abstract}

\section{Focusing and Multifocusing in Linear Logic}

The two most important results in the proof theory of linear logic
\cite{girard:87:ll} are the admissibility of the cut rule, and the
completeness of the focused normal form of proofs. The notion of
\emph{focusing} was originally developped by Andreoli \cite{andreoli:92:foc}
with the purpose of improving proof search procedures, but recently it has
been considered more often as a \emph{normal form} that can be obtained by
reorganising the inference steps of a given proof \cite{miller:saurin:07:focg}:
permutations can be used to group \emph{positive} rule instances --- and
to move \emph{negative} rule instances down. This viewpoint is particularly
useful in the natural extension of focusing to \emph{multifocusing}
\cite{chaudhuri:miller:saurin:08:mfoc}, where several positive formulas
are selected to be decomposed \emph{in parallel}. This stronger normal
form is difficult to use for proof search, since not all positive formulas
can be selected in a given sequent: there are complex \emph{dependencies}.
In the multiplicative fragment without units, when the selection of positives is done
maximally, proofs are canonical in the sense that they are in bijection
with proof-nets \cite{girard:96:pn}. For this reason, investigating the
proof theory of multifocused linear logic is necessary to understand the
notion of canonicity in sequent calculi, and possibly design normal forms
of proofs that could be used as proof-nets, while retaining the usual
syntax based on trees of rule instances.

Our purpose here is to study the interaction of cut elimination with
multifocusing. We consider the sequent calculus \LLM shown in Figure
\ref{figllmsys}, which is equivalent to the one found in \cite{saurin:phd}.
It relies on a \emph{polarised} syntax \cite{laurent:phd} where shifts
mark borders between positive and negative connectives, as follows:
$$\begin{array}{r@{~\grdef~}c@{~\mmid~}c@{~\mmid~}c@{~\mmid~}
    c@{~\mmid~}c@{~\mmid~}c@{~\mmid~}l}
  P,Q & a & \llo & P \otimes Q & \llz & P \oplus Q & \oc N & \pfy N \\
  N,M & \an{a} & \bot & N \parr M & \top & N \with M & \wn P & \nfy P \\
\end{array}$$
and we write $P\lln$ for the usual duality operation of linear logic. Moreover,
this system uses sequents of the shapes $\rseq{\Per:\Gamma}$
and $\fseq{\Per:\Psi}$ for \emph{inversion} and \emph{focusing} phases
respectively, where the names used for multisets denote various syntactic
categories:
$$\begin{array}{r@{~\grdef~}c@{~\mmid~}l@{\qqquad}r@{~\grdef~}c@{~\mmid~}
    l@{\qqquad}r@{~\grdef~}c@{~\mmid~}l}
  \Gamma,\Delta & \cdot & \Gamma,N &
  \Psi,\Xi & \Gamma & \Psi,[P] &
  \Theta,\Per & \cdot & \Theta,P \\
\end{array}$$
and the multiset $\Per$ of positives on the left of the sequent is the
\emph{persistent context}, corresponding to a multiset of formulas that
can be duplicated. Finally, the \emph{decision} rule \rn{\nfy} uses the
special syntax $\Per^{\vec{n}}$ to denote a multiset made of arbitrary
numbers of copies of formulas in $\Per$ --- since this rule needs to
allow choosing formulas from $\Per$ to focus on, and possibly several
copies of the same formula.

\newpage

The most important rules for multifocusing are the \rn{\nfy} and \rn{\pfy}
rules, that start and end the focusing phases, respectively. These two rules
act on multisets of formulas rather than on single formulas, and in particular
it is important that all foci are \emph{blurred} at once in the \rn{\pfy}
rule, so that the different phases are clearly separated. They come with
side conditions:
\begin{itemize}
\item in the \rn{\nfy} rule, either $\Per^{\vec{n}}$ or $\Theta$ must be
      non-empty, and
\item in the \rn{\pfy} rule, the multiset $\Delta$ must be non-empty.
\end{itemize}

\begin{figure}[t]
\centerline{\fbox{
$\begin{array}{c@{\quad}}
  \begin{array}{c@{\qqquad}c@{\qqquad}c}
    \vspace*{-0.5em} \\
    \iaxrule{ax}{\fseq{\Per:\an{a},[a]}} &
    \irule{\pfy}{\fseq{\Per:\Gamma,[\pfy \Delta]}}{\rseq{\Per:\Gamma,\Delta}} &
    \irule{\nfy}{\rseq{\Per:\Psi,\nfy \Theta}}
      {\fseq{\Per:[\Per^{\vec{n}}],\Psi,[\Theta]}} \\
  \end{array} \\
  \\
  \begin{array}{c@{\qquad}c@{\qquad}c@{\qquad}c}
    \irule{\bot}{\rseq{\Per:\Gamma,\bot}}{\rseq{\Per:\Gamma}} &
    \iaxrule{\llo}{\fseq{\Per:[\llo]}} &
    \irule{\parr}{\rseq{\Per:\Gamma,N \parr M}}{\rseq{\Per:\Gamma,N,M}} &
    \iruule{\otimes}{\fseq{\Per:\Psi,\Xi,[P \otimes Q]}}
      {\fseq{\Per:\Psi,[P]}}{\fseq{\Per:\Xi,[Q]}} \\
  \end{array} \\
  \\
  \begin{array}{c@{\qquad}c@{\qqqquad}c}
    \iaxrule{\top}{\rseq{\Per:\Gamma,\top}} &
    \iruule{\with}{\rseq{\Per:\Gamma,N \with M}}
      {\rseq{\Per:\Gamma,N}}{\rseq{\Per:\Gamma,M}} &
    \irule{\wn}{\rseq{\Per:\Gamma,\wn P}}{\rseq{\Per,P:\Gamma}} \\
    \\
    \irule{\oplus_L}{\fseq{\Per:\Psi,[P \oplus Q]}}{\fseq{\Per:\Psi,[P]}} &
    \irule{\oplus_R}{\fseq{\Per:\Psi,[P \oplus Q]}}{\fseq{\Per:\Psi,[Q]}} &
    \irule{\oc}{\fseq{\Per:[\oc N]}}{\rseq{\Per:N}} \\
    \vspace*{-0.5em} \\
  \end{array} \\
\end{array}$}}
\caption{Multifocused sequent calculus \LLM for linear logic}
\label{figllmsys}
\end{figure}

The \LLM system is slightly different from other presentations
\cite{chaudhuri:miller:saurin:08:mfoc,saurin:phd} in that it does not
enforce maximal inversion of negative formulas before a focusing phase
can start. From the viewpoint of cut elimination, it makes no difference,
and proving the admissibility of cut for the more permissive system implies
that cut elimination holds for the corresponding system where inversion
is performed maximally. Moreover, we are interested in developing a proof
technique for admissibility  of cut in multifocused systems that could be
used when negative formulas are treated differently.

\begin{example} \label{ex:prf}
The proof shown below uses multifocusing to treat the formulas
$\nfy (a \otimes \pfy \an{b})$ and $\nfy (c \otimes \pfy \an{d})$ in one
single focusing phase. Notice that the two lower instances of the \rn{\otimes}
rule could be permuted, but this is irrelevant here since a multifocusing
phase should be considered as a ``\,black box\,''. Also, this proof is
maximally multifocused, since in both instances of the \rn{\nfy} rule,
all positive formulas that could be picked are focused.
\vspace*{-0.5em}
$$\ider{
  \idin{\nfy}{\rseq{\cdot:\nfy (a \otimes \pfy \an{b}),\an{a},\an{c},
    \nfy (c \otimes \pfy \an{d}),\nfy (b \otimes d)}}{
  \idiin{\otimes}{\fseq{\cdot:[a \otimes \pfy \an{b}],\an{a},\an{c},
    [c \otimes \pfy \an{d}],\nfy (b \otimes d)}}{
  \idax{ax}{\fseq{\cdot:\an{a},[a]}}}{
  \idiin{\otimes}{\fseq{\cdot:[\pfy \an{b}],\an{c},
    [c \otimes \pfy \an{d}],\nfy (b \otimes d)}}{
  \idax{ax}{\fseq{\cdot:\an{c},[c]}}}{
  \idin{\pfy}{\fseq{\cdot:[\pfy \an{b}],[\pfy \an{d}],\nfy (b \otimes d)}}{
  \idin{\nfy}{\rseq{\cdot:\an{b},\an{d},\nfy (b \otimes d)}}{
  \idiin{\otimes}{\fseq{\cdot:\an{b},\an{d},[b \otimes d]}}{
  \idax{ax}{\fseq{\cdot:\an{b},[b]}}}{
  \idax{ax}{\fseq{\cdot:\an{d},[d]}}}}}}}}}$$
\end{example}

\newpage

Notice finally that the \LLM system is obviously sound with respect to
linear logic, in the sense that all focusing annotations can be erased
to turn any proof of \LLM into a valid proof of a dyadic presentation
of linear logic. It is also complete, since it is a generalisation of
the usual singly-focused system, just as the other multifocused systems
\cite{chaudhuri:miller:saurin:08:mfoc,saurin:phd}.

\section{Cut Elimination with Multifocusing}
Because our system has two different kinds of sequents and two
contexts, we get four different cut rules, shown in
Figure~\ref{figcutllm}. Note that in the linear cut rules, the
first premise must always be a focused sequent, as the cut formula in
this premise is positive, and hence must appear inside a
focus. Conversely, for cuts acting on the persistent context, the
first premise must be an unfocused sequent.

\begin{figure}[t]
\centerline{\fbox{
$\begin{array}{c@{\quad}}
  \begin{array}{c@{\qqquad}c}
    \vspace*{-0.5em} \\
\iruule{cut}{\fpf\Per:\Psi,\Gamma}{\fpf\Per:\Psi,\foc{P}}{\pf\Per:\Gamma,P^\bot}
&
\iruule{fcut}{\fpf\Per:\Psi,\Xi}{\fpf\Per:\Psi,\foc{P}}{\fpf\Per:\Xi,P^\bot}
\end{array}\\
\begin{array}{c@{\qqquad}c}
  \vspace*{-0.5em} \\
  \iruule{cut!}{\pf\Per:\Gamma}{\pf\Per:P^\bot}{\pf\Per,P:\Gamma}
  &
  \iruule{fcut!}{\fpf\Per:\Psi}{\pf\Per:P^\bot}{\fpf\Per,P:\Psi}
\end{array}
\end{array}$
}}
\caption{Cut rules for \LLM}
\label{figcutllm}
\end{figure}

A standard way of proving admissibility of the cut rule is to proceed by
lexicographical induction on the structure of the cut formula and
the two input derivations. The cases of that proof fall into various
categories depending on whether the cut formula is being decomposed (in
the so-called principal cases) or whether it simply moves the cut further
up in the proof (in the so-called commutative cases). Thus, in the
singly-focused system, you might see the following cut:
\vspace{-1.5em}
\[
\ider{
  \idiin{cut}{\pf\Per:\Psi,\Gamma,\pshift\Theta}{
    \idopen{\Dd}{\fpf\Per:\Psi,\foc{P}}
  }{
    \idin{\pshift}{\pf\Per:\Gamma,\pshift\Theta,P^\bot}{
      \idopen{\Ee}{\fpf\Per:\Gamma,\foc{\Theta},P^\bot}
    }
  }
}
\]
And reduce it as follows:
\begin{lineproof}
  \dstepn \fpf\Per:\Psi,\Gamma,\foc{\Theta} & by \rn{fcut} on $P,\Dd,\Ee$.\\
  \dgoal \pf\Per:\Psi,\Gamma,\pshift\Theta & by \rn{\pshift}.\\
\end{lineproof}
In the singly-focused system, $\Psi$ cannot contain any
focus\footnote{Additionally, $\Theta$ must consist of a single formula,
  but this is not important.}, hence the resulting sequent after the
cut contains only the foci in $\Theta$, and thus the \rn{\pshift} rule can
be applied. In the multifocused system, this is no longer the case ---
$\Psi$ may contain several foci, and thus we cannot be sure that the
side condition on the \rn{\pshift} rule is satisfied.

There are various ways one might try to fix the above problem. One way
would be to change the cut rule itself to make the conclusion more
permissive. For instance, if we define the following \emph{neutralising}
operation:
\[
  \delay{\Gamma}\;=\;\Gamma\qqquad
  \delay{\Psi,\foc{P}}\;=\;\delay{\Psi},\pshift P
\]
we can restate the \rn{cut} rules as
\[
\iruule{cut}{\pf\Per:\delay{\Psi,\Gamma}}{\fpf\Per:\Psi,\foc{P}}{\pf\Per:\Gamma,P^\bot}
\qquad
\iruule{fcut}{\pf\Per:\delay{\Psi,\Xi}}{\fpf\Per:\Psi,\foc{P}}{\fpf\Per:\Xi,P^\bot}
\]
and the above proof would then go through, as 
$\delay{\Psi,\Gamma,\foc{\Theta}}=\delay{\Psi},
  \Gamma,\pshift\Theta=\delay{\Psi,\Gamma,\pshift\Theta}.$
The delayed cut above is somewhat weak, however, as it forces a delay
even in cases where it is not necessary.

As we will show now, it is not necessary to change the statement of
the cut rule to prove it admissible. To circumvent the above problem,
we will instead introduce a few lemmas about the structure of the
positive phases. First, we introduce the notion of a \emph{spent} context
$\Sigma$:
\[\Sigma ~\grdef~ \Gamma \mmid \Sigma,\foc{\nshift N}\]

Intuitively, a context $\Sigma$ may contain foci, but none of these
foci can be active anymore. The main purpose of this definition is to
facilitate the proof of the following lemma, which will play an
important role in the cut admissibility proof.
\begin{lemma}[Multifocused Decomposition]
  Given a proof of a sequent $\fpf\Per: \Psi,\foc{P}$ there exists:
  \begin{enumerate}
  \item a proof of the sequent $\fpf\Per:\Sigma,\foc{P}$, for some suitable
    $\Sigma$, and
  \item for any $\Delta$, an open derivation from $\fpf\Per:\Sigma,\Delta$ to
    $\fpf\Per:\Psi,\Delta$.
  \end{enumerate}
  The combined height of these two derivations is exactly the height
  of the input derivation.
\end{lemma}

\begin{proof}
  Let $\Dd$ be the given derivation of $\fpf\Per:\Psi,\foc{P}$. We proceed by
  induction on the structure of $\Dd$. When we apply the induction
  hypothesis in the remainder of this proof, we will refer to the
  first derivation as $\Dd'$, and the open derivation instantiated with
  $\Delta$ as $\Ee'_\Delta$.
  \begin{description}
    \item[Case $\Psi=\Sigma$:] 
      Immediate. This case also covers the cases where $\Dd$ ends in
      the \rn{\llo}, \rn{ax} or \rn{\nshift} rules.
    \item[Case \rn{\tensor}, $P=Q\tensor R$ principal:] $\phantom{oe}$
      \vspace{-2em}
      \[\ider{\idiin{\tensor}{\fpf\Per:\Psi_1,\Psi_2,\foc{Q\tensor R}}
                        {\idopen{\Dd_1}{\fpf\Per:\Psi_1,\foc{Q}}}
                        {\idopen{\Dd_2}{\fpf\Per:\Psi_2,\foc{R}}}}\]
      We construct $\Dd'$ as follows:
      \begin{lineproof}
        \dstep \Dd_1' :: \fpf\Per: \Sigma_1,\foc{Q} & by the induction hypothesis on $\Dd_1$.\\
        \dstep \Dd_2' :: \fpf\Per: \Sigma_2,\foc{R} & by the induction hypothesis on $\Dd_2$.\\
        \dgoal \fpf\Per: \Sigma_1,\Sigma_2,\foc{Q\tensor R} & by \rn{\tensor} on
        $\Dd_1'$, $\Dd_2'$.\\
      \end{lineproof}
      We construct $\Ee'_\Delta$ as follows:
      \begin{lineproof}
        \dstepn \fpf\Per:\Sigma_1,\Sigma_2,\Delta & by assumption.\\
        \dstepn \fpf\Per:\Sigma_1,\Psi_2,\Delta & by ${\Ee_2'}_{\Delta,\Sigma_1}$.\\
        \dgoal \fpf\Per:\Psi_1,\Psi_2,\Delta & by ${\Ee_1'}_{\Delta,\Psi_2}$.\\
      \end{lineproof}
    \item[Case \rn{\tensor}, $P$ not principal:] $\phantom{oe}$
      \vspace{-2em}
      \[\ider{\idiin{\tensor}{\fpf\Per:\Psi_1,\Psi_2,\foc{Q\tensor R},\foc{P}}
      {\idopen{\Dd_1}{\fpf\Per:\Psi_1,\foc{Q},\foc{P}}}
      {\idopen{\Dd_2}{\fpf\Per:\Psi_2,\foc{R}}}}
       \]
      We construct $\Dd'$ as follows:
      \begin{lineproof}
        \dgoal \fpf\Per:\Sigma,\foc{P} & by the induction hypothesis on $\Dd_1$.\\
      \end{lineproof}

      \newpage

      We construct $\Ee'_\Delta$ as follows:
      \begin{lineproof}
        \dstepn \fpf\Per:\Sigma,\Delta & by assumption.\\
        \dstep \Dd_1'' ::  \fpf\Per:\Psi_1,\foc{Q},\Delta & by ${\Ee_1'}_\Delta$.\\
        \dgoal \fpf\Per:\Psi_1,\Psi_2,\foc{Q\tensor R},\Delta & by \rn{\tensor} on
        $\Dd_1''$ and $\Dd_2$.\\
      \end{lineproof}
  \end{description}
  The remaining cases are similar.
\end{proof}

In some cases, we will make use of the following observation about
sequents of the form $\fpf\Per:\Sigma$. First we will define a
\emph{lowering} operation that inverts spent foci:
\def\lowered#1{\lfloor#1\rfloor}
\[
 \lowered{\Gamma}\;=\;\Gamma\qqquad
 \lowered{\Sigma,\foc{\nshift N}}\;=\;\lowered{\Sigma},N
\]
With the above definition, the following lemma is an easy consequence:
\begin{lemma}[Lowering spent foci]
  The following rule is admissible:
  \[\irule{\nshift^{-1}}{\fpf\Per:\lowered{\Sigma},\foc{P}}{\fpf\Per:\Sigma,\foc{P}}\]
  Furthermore, it is \emph{strongly} admissible, in the sense that
  applying the rule does not change the shape of the resulting derivation.
\end{lemma}
\begin{proof}
  By induction on the given derivation of $\fpf\Per:\Sigma,\foc{P}$. The
  crucial observation is the fact that since $\Sigma$ only contains
  spent foci, the only active formula is $P$. The base case is when
  $P=\pfy N$, and in this case the only rule that could produce the
  sequent $\fpf\Per:\Sigma,\foc{\pfy N}$ is the \rn{\pfy} rule applied to the
  sequent $\pf\Per:\lowered{\Sigma},N$. By applying the \rn{\pfy} rule again,
  but this time only to the formula $N$, we get the desired conclusion.
\end{proof}

With the above lemmas and rules, we can now tackle the admissibility
of the cut rules:
\begin{theorem}[Admissibility of cut]
  The rules in Figure~\ref{figcutllm} are admissible in \LLM.
\end{theorem}
\begin{proof}
  By lexicographic induction on the cut formula and the derivations. We show
  here a representative selections of cases, in particular some of the principal
  and boundary cases of the linear cuts.  
  Because we have the decomposition lemma, we can use it on the first
  subderivation of each rule. It is therefore sufficient to prove the
  admissibility of the following {\it ``spent cut''} rules:
  \[
  \iruule{scut}{\fpf\Per:\Sigma,\Gamma}{\fpf\Per:\Sigma,\foc{P}}{\pf\Per:\Gamma,P^\bot}
  \qquad\qquad
  \iruule{fscut}{\fpf\Per:\Sigma,\Psi_2}{\fpf\Per:\Sigma,\foc{P}}{\fpf\Per:\Psi_2,P^\bot}
  \]
  By the following transformation:
  \vspace{-1.5em}
      \[\ider{\idiin{cut}{\fpf\Per:\Psi,\Gamma}
                        {\vlde{\Ee'_{\foc{P}}}{}{\fpf\Per:\Psi,\foc{P}}{\idopen{\Dd'}{\fpf\Per:\Sigma,\foc{P}}}}
                        {\idopen{\Ee}{\pf\Per:\Gamma,P^\bot}}}
      \qquad\leadsto\qquad
      \ider{\vlde{\Ee'_{\Gamma}}{}{\fpf\Per:\Psi,\Gamma}{\idiin{scut}{\fpf\Per:\Sigma,\Gamma}{
            \idopen{\Dd'}{\fpf\Per:\Sigma,\foc{P}}}
           {\idopen{\Ee}{\pf\Per:\Gamma,P^\bot}}}}
      \]
      
      \newpage

  \noindent
  In fact, we can always expand the \rn{scut} rule as follows:
  \vspace{-2.5em}
  \[\ider{\idiin{scut}{\fpf\Per:\Sigma,\Gamma}
  {\idopen{\Dd}{\fpf\Per:\Sigma,\foc{P}}}
  {\idopen{\Ee}{\pf\Per:\Gamma,P^\bot}}
    }
  \qquad\leadsto\qquad
  \ider{
    \idin{\nshift}{\fpf\Per:\Sigma,\Gamma}{
      \idiin{\lowered{scut}}{\pf\Per:\lowered{\Sigma},\Gamma}
        {\idin{\nshift^{-1}}{\fpf\Per:\lowered{\Sigma},\foc{P}}{\idopen{\Dd}{\fpf\Per:\Sigma,\foc{P}}}}
        {\idopen{\Ee}{\pf\Per:\Gamma,P^\bot}}
      }
    }\]
  This is justified because the \rn{\nshift^{-1}} rule is strongly
  admissible. It is thus sufficient to show admissibility of the
  $\rnm{\lowered{scut}}$ rule.
  For the principal cuts we reason as follows:
  \vspace{-1.5em}
  \[
  \ider{
    \idiin{\lowered{scut}}{\pf\Per:\lowered{\Sigma_1},\lowered{\Sigma_2},\Gamma}
    {\idiin{\tensor}{\fpf\Per:\lowered{\Sigma_1},\lowered{\Sigma_2},\foc{P\tensor
          Q}}
      {
        \idopen{\Dd_1}{\fpf\Per:\lowered{\Sigma_1},\foc{P}}
      }
      {
        \idopen{\Dd_2}{\fpf\Per:\lowered{\Sigma_2},\foc{Q}}
      }}
    {\idin{\parr}{\pf\Per:\Gamma,P^\bot\parr Q^\bot}{
      \idopen{\Ee}{\pf\Per:\Gamma,P^\bot,Q^\bot}}}
  }\]
\begin{lineproof}
  \dstep \Ee' :: \pf\Per:\lowered{\Sigma_1},\Gamma,Q^\bot & by
  \rn{\lowered{scut}} on $P, \Dd_1,\Ee$.\\
  \dgoal \pf\Per:\lowered{\Sigma_1},\lowered{\Sigma_2},\Gamma & by
  \rn{\lowered{scut}} on $Q, \Dd_2,\Ee'$.\\
\end{lineproof}
\vspace{-2em}
\[
\ider{
  \idiin{\lowered{scut}}{\pf\Per:\lowered{\Sigma},\Gamma,\pshift\Theta}{
    \idin{\nshift}{\fpf\Per:\lowered{\Sigma},\foc{\nshift P^\bot}}{
      \idopen{\Dd}{\pf\Per:\lowered{\Sigma},P^\bot}
    }
  }{
    \idin{\pshift}{\pf\Per:\Gamma,\pshift\Theta,\pshift P}{
      \idopen{\Ee}{\fpf\Per:\foc{\Per^{\vec n}},\Gamma,\foc{\Theta},\foc{P}}
    }
  }
}
\]
\begin{lineproof}
  \dstepn \fpf\Per:\lowered{\Sigma},\foc{\Per^{\vec n}},\Gamma,\foc{\Theta} & by \rn{cut} on $P,\Ee,\Dd$.\\
  \dgoal \pf\Per:\lowered{\Sigma},\Gamma,\pshift\Theta & by \rn{\pshift}.\\
\end{lineproof}
Note that we here appeal to the fully general \rn{cut} rule rather
than a more specific rule. In doing so, we implicitly apply the
decomposition highlighted at the beginning of this proof. 

The commutative cuts are dispensed with in a similar manner. 
We show here a few instances:
  \[
  \ider{
    \idiin{\lowered{scut}}{\pf\Per:\lowered{\Sigma},\Gamma,N\parr M}
    {\idopen{\Dd}{\fpf\Per:\lowered{\Sigma},\foc P}}
    {
      \idin{\parr}{\pf\Per:\Gamma,P^\bot,N\parr M}{
        \idopen{\Ee}{
          \pf\Per:\Gamma,P^\bot,N, M
        }
      }}
  }\]
\begin{lineproof}
  \dstep \Ee' :: \pf\Per:\lowered{\Sigma},\Gamma,N,M & by
  \rn{\lowered{scut}} on $P, \Dd,\Ee$.\\
  \dgoal \pf\Per:\lowered{\Sigma},\Gamma,N\parr M & by
  \rn{\parr}.\\
\end{lineproof}

  \[
  \ider{
    \idiin{fcut}{\fpf\Per:\Psi,\Phi}
    {\idopen{\Dd}{\fpf\Per:\Psi,\foc P}}
    {
      \idin{\oplus_L}{\fpf\Per:\Phi,\foc{Q\oplus R},P^\bot}{
        \idopen{\Ee}{
          \pf\Per:\Phi,\foc{Q},P^\bot,
        }
      }}
  }\]
\begin{lineproof}
  \dstep \Ee' :: \fpf\Per:\Psi,\Phi,\foc{Q} & by
  \rn{fcut} on $P, \Dd,\Ee$.\\
  \dgoal \fpf\Per:\Psi,\Phi,\foc{Q\oplus R} & by
  \rn{\oplus_L}.\\
\end{lineproof}

When the second derivation is at the
boundary of a focusing phase, we have restrictions on the shape of the
context. In this case, the fact that we can restrict the shape of the
context of the first premise of the cut rule becomes crucial. 
\vspace{-1.5em}
\[
\ider{
  \idiin{fscut}{\fpf\Per:\lowered{\Sigma},\Gamma,\foc{\nshift\Gamma'}}{
    \idopen{\Dd}{\fpf\Per:\lowered{\Sigma},\foc{P}}
  }{
    \idin{\nshift}{\fpf\Per:\Gamma,\foc{\nshift\Gamma'},P^\bot}{
      \idopen{\Ee}{\pf\Per:\Gamma,\Gamma',P^\bot}
    }
  }
}
\]
\begin{lineproof}
  \dstepn \pf\Per:\lowered{\Sigma},\Gamma,\Gamma' & by \rn{\lowered{scut}} on $P,\Dd,\Ee$.\\
  \dgoal \fpf\Per:\lowered{\Sigma},\Gamma,\foc{\nshift\Gamma'} & by \rn{\nshift}.\\
\end{lineproof}
\vspace{-1.5em}
\[
\ider{
  \idiin{\lowered{scut}}{\pf\Per:\lowered{\Sigma},\Gamma,\pshift\Theta}{
    \idopen{\Dd}{\fpf\Per:\lowered{\Sigma},\foc{P}}
  }{
    \idin{\pshift}{\pf\Per:\Gamma,\pshift\Theta,P^\bot}{
      \idopen{\Ee}{\fpf\Per:\foc{\Per^{\vec n}},\Gamma,\foc{\Theta},P^\bot}
    }
  }
}
\]
\begin{lineproof}
  \dstepn \fpf\Per:\lowered{\Sigma},\foc{\Per^{\vec n}},\Gamma,\foc{\Theta} & by \rn{fscut} on $P,\Dd,\Ee$.\\
  \dgoal \pf\Per:\lowered{\Sigma},\Gamma,\pshift\Theta & by \rn{\pshift}.\\
\end{lineproof}

\end{proof}

\section{Further Restrictions on Multifocused Systems}

As we have seen, it is possible to prove cut elimination in a sensible,
internal way in the multifocused sequent calculus \LLM. However, further
restrictions on proofs, designed to obtain stronger normal forms, could
yield problems. The obvious restriction to consider here is the one that enforces
the \emph{maximality} of multifocusing \cite{chaudhuri:miller:saurin:08:mfoc}
--- note that such a notion of maximality can be observed on proofs but it is
difficult to provide a syntax enforcing maximality. But in this setting, we
encounter a problem with the following configuration, for example:
\vspace*{-1.5em}
\begin{equation} \label{eq:badcut}
\ider{
  \idiin{cut}{\fseq{\cdot:\nfy (a \otimes \pfy \an{b}),[\pfy \bot],
    \an{a},\an{c},\nfy (c \otimes \pfy \an{d}),\nfy (b \otimes d)}}{
  \idtrees{\Dd}{\fseq{\cdot:\nfy (a \otimes \pfy \an{b}),[\pfy \bot],
    [\pfy (\an{a} \parr \nfy b)]}}{2em}}{
  \idtrees{\Ee}{\qquad\rseq{\cdot:\nfy (a \otimes \pfy \an{b}),\an{a},\an{c},
    \nfy (c \otimes \pfy \an{d}),\nfy (b \otimes d)}}{2em}}}
\end{equation}
where the proof $\Dd$ is a simple variation of \emph{identity expansion},
and $\Ee$ is the proof given in Example \ref{ex:prf}, both of them being
maximally multifocused proofs. Now if we want to perform cut elimination
in this restricted system, we expect the resulting proof to be maximally
multifocused as well. However, this will separate the phase treating
$\pfy \bot$ from the lower focusing phase in $\Ee$, although these could
be merged into a single phase. Therefore, if this situation arises during
cut elimination, with some derivation below that concludes in an unfocused
sequent, then the result will not be a maximally multifocused proof. In
the presentation of \cite{chaudhuri:08:synth}, there is a proof
that certain cut elimination strategies preserve maximality. However, this
is done in the higher-order focusing style due to Zeilberger
\cite{zeilberger:08:unidual}, and this only applies to the MALL fragment of
linear logic.

At this point, a better understanding of the dynamics of cut elimination
in this setting is required, and in particular we need to control the
\emph{merging} of focusing phases, so that no pair of phases is left that
could be turned into a single multifocusing phase. A change in the statement
of cut admissibility might be required, or possibly a more elaborate proof
technique, potentially involving rewritings even less {\it ``local''}
than the ones used in the previous section. In particular, one can consider
the following (non-deterministic) \emph{activating} operation dual to the
neutralising one:
$$\begin{array}{c@{\qqquad}c@{\qqquad}c@{\qqquad}c}
    \iaxrule{}{\cdot \activ \cdot} &
    \irule{}{\Psi,N \activ \Xi,N}{\Psi \activ \Xi} &
    \irule{}{\Psi,\nfy P \activ \Xi,[P]}{\Psi \activ \Xi} &
    \irule{}{\Psi,[P] \activ \Xi,[P]}{\Psi \activ \Xi} \\
\end{array}$$
and the following kind of cut, where we have $\Psi,\Gamma \activ \Xi$:
$$\iruule{acut}{\fseq{\Per:\Xi}}{\fseq{\Per:\Psi,[P]}}{\rseq{\Per:\Gamma,P\lln}}$$
so that the conclusion can have \emph{more} foci than either of the
premises. Such a cut would potentially allow the cut shown in
\eqref{eq:badcut} to reduce properly, since it offers the possibility
to have the multifocused sequent $\fseq{\cdot:[a \otimes \pfy \an{b}],
[\pfy \bot],\an{a},\an{c},[c \otimes \pfy \an{d}],\nfy (b \otimes d)}$
as a conclusion, for which a maximally multifocused proof exists ---
yielding a maximally multifocused proof of the corresponding unfocused
sequent. It suggests that internal cut elimination might be possible in
this setting, for an adequate cut.

Alternatively, one could allow foci to be present during the inversion phase,
by replacing the $\nfy$ rule with the following three rules:
\[
  \irule{\mathrm{focus}}{\rseq{\Per:\Psi}}
      {\fseq{\Per:\Psi}}\qquad
  \irule{\mathrm{copy}}{\rseq{\Per,P:\Psi}}
      {\rseq{\Per,P:\Psi,\foc{P}}}\qquad
  \irule{\nfy}{\rseq{\Per:\Psi,\nfy P}}
      {\rseq{\Per:\Psi,\foc P}}
\]
With these rules, it becomes possible to select formulas for focusing one at
a time during the inversion phase. This is a somewhat radical step away from
the usual presentations of focusing, and so care must be taken to ensure that
the resulting system is still well-behaved. With these new rules, the cut rule
can now be written as follows:
\[
\iruule{cut}{\pf\Per:\Psi,\Phi}{\fpf\Per:\Psi,\foc{P}}{\pf\Per:\Phi,P^\bot}
\]
Note the two differences in comparison to the rules presented previously.
First of all, the context in the second premise can now contain additional
foci, and is thus of the form $\Phi$ rather than $\Gamma$. Secondly, the
conclusion is in the judgment corresponding to the inversion phase. With
these changes, it is no longer necessary to have a concept of spent cuts,
and the proof is greatly simplified.

It would be tempting to adopt the following symmetric rules instead of
the $\pfy$ rule:
\[
  \irule{\mathrm{blur}}{\fseq{\Per:\Gamma}}
      {\rseq{\Per:\Gamma}}\qquad
  \irule{\pfy}{\fseq{\Per:\Psi,\foc{\pfy N}}}
      {\fseq{\Per:\Psi,N}}
\]
Observe that in the context in the blur rule, no foci can be present. 
This is to prevent foci from ``bleeding through'' from one phase to the
next. Here, however, caution must be exercised, as exchanging the above
rules yields a system that fails to enjoy the cut elimination property.

Finally, another interesting restriction of the multifocused system could be defined
by controlling inversion phases on negatives in a finer way: instead of
maximally decomposing negatives, deal only with the ones that will be
{\it ``needed''} within the next multifocusing phase. Fortunately, the
inversion of negatives is not critical in the cut elimination proof, so
that such a system would have a strong notion of \emph{bipole} and yet
use the techniques used in the previous section for cut elimination.

\section{Conclusion and Future Work}

We presented here a proof of cut elimination for a multifocused sequent
calculus for linear logic, crucially relying on a decomposition lemma
that expresses the fact that parallel phases can be permuted with one
another inside a multifocused phase. This kind of order irrelevance is
the essence of focusing and, at the {\it ``higher''} level of focusing
phases, this is also the essence of multifocusing: we believe such a lemma
can be a useful tool when proving cut elimination for even stronger
restrictions of the sequent calculus. Also, operations such as neutralising
or activating might prove useful if the precise statement of cut needs
to be changed to fit the restrictions of a system.

However, many questions remain open, such as the precise interaction
between cut elimination and the organisation of multifocusing phases in
a proof, or the definition of a maximally multifocused system equipped
with an adequate cut --- which could be eliminated entirely inside the
system. From a broader perspective, the question of strong normal forms
in the sequent calculus or other related systems allowing for permutations is tied
to our ability to perform cut elimination in a sensible way under the
constraints imposed by multifocusing or even stronger restrictions. For this
reason it would be interesting to consider multifocusing in settings
such as linear natural deduction \cite{brocknannestad:schuermann:10:focnd}
or the calculus of structures \cite{chaudhuri:guenot:strassburger:11:foccos}.
From the viewpoint of computation, this means investigating cut elimination
as a form of computation in a system where many interleavings of independent
steps are abstracted away, as done with proof-nets. It connects the structural
approach of standard proof theory to the study of graph-based computational
models \cite{guerrini:martini:masini:01:pngc}, but this still requires to
improve the understanding of normal forms for larger fragments of linear
logic.

Finally, extending the notion of focusing raises the question of the
elegance of completeness proofs for multifocused systems. Indeed, it seems
difficult to prove the focusing result through cut elimination in a simple,
natural way \cite{simmons:14:strfoc} if the cut elimination proof itself is
complex.

\vspace*{1em}
{\bf Acknowledgements}. This work was partially funded by the \emph{Demtech}
grant number 10-092309 from the Danish Council for Strategic Research.

\begin{raggedright}
\bibliographystyle{eptcs/eptcs.bst}
\bibliography{cite}
\end{raggedright}

\end{document}